\documentclass[10pt]{article}

\usepackage{amsmath}
\usepackage{amsthm}
\usepackage{amssymb}
\usepackage{hyperref}

\usepackage{graphicx}

\newtheorem{thm}{Theorem}

\newtheorem{lemma}[thm]{Lemma}
\newtheorem{prop}[thm]{Proposition}

\theoremstyle{remark}
\newtheorem{remark}{Remark}

\theoremstyle{definition}

\newcommand{\R}{\mathbb R}

\newcommand{\Rp}{{\mathbb R}_+}

\newcommand{\cL}{\mathcal{L}}
\newcommand{\cM}{\mathcal{M}}

\newcommand{\g}{{\gamma}}

\newcommand{\lam}{{\lambda}}           

\newcommand{\Om}{\Omega}                

\newcommand{\s}{{\sigma}}

\newcommand{\ran}{\rangle}
\newcommand{\lan}{\langle}
\newcommand{\ra}{\rightarrow}

\newcommand{\p}{{\partial}}

\newcommand{\lb}{\left(}
\newcommand{\rb}{\right)}
\newcommand{\lsb}{\left[}
\newcommand{\rsb}{\right]}

\newcommand{\const}{\operatorname{const}}

\newcommand{\ip}[2]{\left\langle #1,#2\right\rangle}

\newcommand{\Rin}{R_{i}}
\newcommand{\Xout}{X_{o}}
\newcommand{\Xin}{X_{i}}
\newcommand{\Rout}{R_{o}}
\newcommand{\Lin}[1]{\left\|#1\right\|_{i}}
\newcommand{\Lout}[1]{\left\|#1\right\|_{o}}

\newcommand{\Lab}{{\cal L}_{a}}
\newcommand{\La}{L_{a}}
\newcommand{\Fab}{F_{a}}
\newcommand{\Fabc}{F_{a b c}}
\newcommand{\Nab}{N}
\newcommand{\Lap}[1]{\Delta^{\!\!(#1)}}
\newcommand{\Labc}{{\cal L}_{abc}}
\newcommand{\epsin}{\varepsilon_i}
\newcommand{\epsout}{\varepsilon_o}
\newcommand{\Lopout}{{\cal L}_o}
\newcommand{\Vout}{V_o}
\newcommand{\Lopin}{{\cal L}_i}
\newcommand{\Vin}{V_i}
\newcommand{\Va}{V}
\newcommand{\Wa}{W_a}

\renewcommand{\O}[1]{\mathrm{O}\lb#1\rb}
\newcommand{\OX}[2]{\mathrm{O}_{#1}\lb#2\rb}
\newcommand{\smallO}[2]{\mathrm{o}_{#1}\lb#2\rb}

\newcommand{\DETAILS}[1]{}

\title{On Spectra of Linearized Operators for Keller-Segel Models of Chemotaxis\thanks{Supported by NSF grants DMS 0719895 and
DMS 0807131,  UNM RAC grant, RFBR  grant 
RFBR-07-0-2-12058, and NSERC grant 7901 }}
\author{S. I. Dejak\footnote{Department of Mathematics, University of Toronto, Toronto, Canada.}\,\,\ P.M. Lushnikov\footnote{Department of Mathematics and Statistics, University of New Mexico, USA.}\,\,\  Yu. N. Ovchinnikov\footnote{Landau Institute for Theoretical Physics, Chernogolovka, Russia and Max-Planck Institute for Physics of Complex Systems, 01187 Dresden, Germany.}\,\,\   I. M. Sigal$^\dagger$}
\date{October 12, 2011}
\newcommand{\DATUM}{September 21, 2011}              
\begin{document}

\maketitle

\begin{abstract}
We consider  the phenomenon of collapse in the critical  Keller-Segel equation (KS) which models chemotactic aggregation of micro-organisms  underlying many social activities, e.g. 
fruiting body development 
 and biofilm formation. Also  KS describes the collapse of a gas of self-gravitating Brownian particles.
We find the fluctuation spectrum around the 
 collapsing  family of steady states for these equations, which is instrumental in derivation of  the critical collapse law.
To this end we develop a rigorous version of the method of  matched asymptotics for the spectral analysis of a class of second order differential operators containing the linearized Keller-Segel operators
(and as we argue  linearized operators appearing in nonlinear evolution problems). 
 We explain how the
results we obtain are used to derive the critical collapse law,  as well as for proving its stability.
 \end{abstract}
\textit{Key words}: Matched asymptotics, critical Keller-Segel equation, collapse and formation of singularities, linearized operators.

\section{Introduction}

Phenomena of blowup and collapse in nonlinear evolution equations 
are hard to simulate numerically and the rigorous theory, or at least a careful analysis, is  pertinent here.
The recent years witnessed a tremendous progress in developing of such theories. 
We can now describe  the shape of blowup profile and contraction law 
in Yang-Mills, $\s-$model,  nonlinear Schr\"odinger and heat  equations (\cite{RS, RR, KST1, KST2, OS, BOS, MR, MZ, DGSW})
\footnote{Numerical simulations for these equations failed until the compression rate was derived analytically, see \cite{BOS, OS, SS}}.
Yet, after 40 years of intensive research and important progress, we still cannot give a rigorous description of collapse in the Keller-Segel equations modeling chemotaxis.
(See \cite{BrLeBu1998,Ve1, Ve2, BCC, BCL, BCM, BDEF, BDP, BeCL} for some recent works, \cite{BrCoKaScVe1999}, for a nice discussion of the subject,  
and \cite{Na2001, Ho2003, Ho2004, HP, Per} for reviews.)

This is not say that  the Keller-Segel equations are harder than  Yang-Mills, $\s-$model and nonlinear Schr\"odinger equations, they are not, but neither are they less important.
They model 
chemotaxis, which is a directed movement of organisms in response to the concentration gradient of an external chemical signal (\cite{KeSe1970}, see also \cite{Patlak1953}). 
The chemical signals can come from external sources or they can be produced by the organisms themselves.  The latter situation leads to aggregation of organisms and to the formation of patterns and is
 the case modeled by the Keller - Segel equations. 
Chemotaxis is believed to underly many social activities of micro-organisms, e.g. social motility, fruiting body development, quorum sensing and biofilm formation. A classical example
is the dynamics and the aggregation of {\it Escherichia coli} colony under the starvation condition \cite{BrLeBu1998}. 
Another example is the   {\it Dictyostelium} amoeba , where single cell bacterivores,
 when challenged by adverse conditions, form multicellular structures of $\sim 10^5$ cells \cite{Bo1967,CFTV}.  Also endothelial cells of humans react to vascular endothelial growth factor to form blood vessels
 through aggregation \cite{CarmelietNatMEd2000}.

We assume that the organism population is large and the individuals are small relative to both the domain, $\Omega\subset \R^d\ (d=1, 2, 3)$ as well as
the typical distance between the organism is much larger than their size.  One can derive in the mean-field approximation
the Keller-Segel system  governing the organism density $\rho:\Omega\times\Rp\rightarrow\Rp$ and chemical
concentration $c:\Omega\times\Rp\rightarrow\Rp$ \cite{KeSe1970,NewmanGrima2004}.
As the chemical diffuses much faster than organisms, one makes a simplifying assumption of instantaneous interaction (adiabatic assumption)
which, after rescaling and a minor simplification, leads the  Keller-Segel equations to the form
\begin{align}\label{KS}
\begin{cases}
 \p_t \rho&=\Delta\rho-\nabla\cdot(\rho\nabla c),\\
0&=\Delta c+\rho,
\end{cases}
\end{align}
with $\rho$ and $c$ satisfying the no-flux Neumann boundary conditions.  The equations \eqref{KS} 
appear also in the context of stellar collapse (see \cite{HeNeVe1997, Wo, ChavSir, SirChav});  similar equations - the Smoluchowski or nonlinear Fokker - Planck equations - models non-Newtonian
complex fluids (see \cite{Doi, Lar, CKT1, CKT2}.

Arguably, the most interesting feature of 
the  Keller-Segel equations is that they can develop, in finite time, infinite mass at a point in space. As argued below, the 'collapsing' profile and contraction law have a universal (close to self-similar)
 form, independent of particulars of initial configurations and, to a certain degree, of the equations themselves, and can be associated with chemotactic aggregation.
 Though the equations are rather crude and unlikely to produce  patterns one observes in nature or experiments, the collapse phenomenon could be useful in  verifying
 assumptions about biological mechanisms.\footnote{There are numerous refinements of the Keller - Segel equations, e.g. taking into account finite size of organisms
 (\cite{AlberChenGlimmLushnikov2006,AlberChenLushnikovNewman2007,LushnikovChenalberPRE2008}) preventing
 the complete collapse, which model the chemotaxis more precisely. We believe techniques we outline and develop here can be applied to these models as well.}

We now concentrate on the (energy) critical case of $d=2$ and $\Om=\R^2$. It was shown in \cite{Na1995, Bi1998} that solutions of \eqref{KS} with the mass
$$M:=\int_{\R^2} \rho_0\, dx>8\pi$$ blow up in finite time. Ref. \cite{Ve2} exhibited blowup solutions with explicit blowup rate and explicit asymptotics, which was confirmed in \cite{Lu,DyachenkoLushnikovVladimirovaAIPConfProc2011}
 by a different
technique relying on results of the present paper.
However, the problem of describing the dynamics of blowup, i.e. blowup rate and profile for an open set of initial conditions is still open. As is shown below, this paper makes a considerable progress toward its solution.

  Of a critical importance here are the following key properties of the equation \eqref{KS}:
\begin{itemize}
\item It is invariant under the scaling transformations
$\rho(x,t)\rightarrow\frac{1}{\lambda^2}\rho\lb\frac{1}{\lambda}x,\frac{1}{\lambda^2}t\rb\
\mbox{and}\ c(x,t)\rightarrow c\lb\frac{1}{\lambda} x,\frac{1}{\lambda^2} t\rb.$

\item
It has the static solution,
$R(x):=\frac{8}{(1+|x|^2)^2},\ C(x):=-2\ln(1+|x|^2).$ 

\item  The total 'mass' is conserved:
$\int \rho(x, t)dx = \const.$\footnote{Another important property of \eqref{KS} that it is a gradient system with the (free) energy $F(\rho)= \int_{\R^2} (
\frac{1}{2}\rho\ \Delta^{-1}\rho+\rho\ln\rho)\, dx$, which plays the key role in other papers, is not used in our approach.} 
\end{itemize}
\DETAILS{Of a critical importance here is the existence of the family of static solutions, $R_\lam(x):=\lam^{-2}R(x/\lam),\ R(x)=\frac{8}{(1+|x|^2)^2}$, of \eqref{KS}.}
The stationary solution $R(x)$ has the total mass $\int R(x) dx=8\pi$, which 
 is exactly the sharp threshold between global existence and singularity development in solutions to \eqref{KS}, mentioned above. 

 The properties above yield that \eqref{KS} has in fact the family of static solutions 
$\lam^{-2}R(x/\lambda),\  C(x/\lambda),\ \lambda>0,$\footnote{It seems this family was discovered in \cite{Ost}. It is shown in \cite{DLS} that belongs to the two parameter family
$R^{(\mu)}_{\lambda}(x):=R^{(\mu)}(r/\lam)$, where $R^{(\mu)}(x):=2(\mu-2)^2\frac{|x|^{\mu-4}
}{(1+|x|^{\mu-2})^2},\  \mu> 2.$ Our case is $\mu=4$. If $2<\mu < 4$, then the mass at the origin is non-zero, 
and if $\mu>4$, then the mass at the origin is negative and hence the static solution is not physical. For $\mu=4$, due to the sharp logarithmic Hardy-Littlewood-Sobolev inequality, these static solutions unique and minimize the free energy, $F(\rho)= \int_{\R^2} (
\frac{1}{2}\rho\ \Delta^{-1}\rho+\rho\ln\rho)\, dx$, for the fixed mass $\int\rho=\mu$ (\cite{B, CL}).  We conjecture that the same is true for $2<\mu < 4$.}
and suggest a likely scenario of collapse: sliding along this family in the direction of $\lambda \ra 0$.
Indeed, we \textit{conjecture} that, like 
in Struwe's result \cite{Str} for equivariant  wave maps from the Minkovskii  space-time, ${\mathbb{M}}^{2+1}$, to the $2$-sphere, $S^2$,  
for any  solution,  $\rho(x, t)$, of \eqref{KS}, collapsing up at time $T$, there are sequences $\lam_i\ra 0$ and $t_i\ra T$, s.t. $\rho(\lam_i y, t_i)$ converges to the stationary solution $R(y)$, as $i\ra \infty$. Thus
the most interesting and natural initial conditions for \eqref{KS} are those close to the manifold $\{R_\lam(x) | \lam>0\}$. 

This discussion brings us to the first step 
of the theory of collapse in the Keller - Segel system - determining the low-lying spectrum of fluctuations around the family $R_\lam(x)$. This would determine whether this family is stable. 
In this paper we find this spectrum and to do this we develop a rigorous version of the method of matched asymptotics.

Now, we discuss,  following \cite{DLS}, a natural approach to this problem. 
Since the blowup profile is expected to be radially symmetric, it is natural to start with radially symmetric solutions. In this case,  the system  \eqref{KS}, which consists of  coupled parabolic and elliptic PDEs, 
is equivalent to a single PDE. Indeed, the change of the unknown, by passing from the density, $\rho(x,t)$, to the normalized mass, 
\begin{equation*}
m(r,t):=\frac{1}{2\pi}\int_{|x|\le r}\rho(x,t)\ dx,
\end{equation*}
  of organisms contained in a ball of radius $r$, 
discovered by  \cite{JaLu1992, BrCoKaScVe1999}, maps two equations \eqref{KS} into a single equation  
\begin{equation}\label{m-eq}
\p_t m=\Lap{0}_r m+r^{-1} m\p_r m,
\end{equation}
on $(0,\infty)$ (with initial condition $m_0(r):=\frac{1}{2\pi}\int_{|x|\le r}\rho_0 (x)\, dx$).  Here $\Lap{n}_r$ is the $n$-dimensional radial Laplacian, $\Lap{n}_r:=r^{-(n-1)}\p_r r^{n-1}\p_r=\p_r^2+\frac{n-1}{r}\p_r$. 
Thus \eqref{KS} in the radially symmetric case is equivalent to \eqref{m-eq}.

The equation \eqref{m-eq} has the following key properties, inherited from the corresponding
properties of \eqref{KS}):
\begin{itemize}
\item It is invariant under the scaling transformations
$m(r,t)\rightarrow m\lb\frac{1}{\lambda}r,\frac{1}{\lambda^2}t\rb.$ 
\item
It has the static solution 
$ \chi(r):=\frac{4 r^2}{1+r^2}$. 

\item  The total 'mass' is conserved:
$2\pi \lim_{r\ra\infty} m(r,t)=\int \rho(x, t)dx = \const.$ 
\end{itemize}  

\DETAILS{ The stationary solution $\chi(r)$ has the total mass $2\pi \lim_{r\ra\infty} \chi(r)=\int R(x) dx=8\pi$, which turns out to be the sharp threshold between global existence and singularity development in solutions to \eqref{KS} or \eqref{m-eq} (see 
 \cite{Na1995, Bi1998}).}
As in the case of \eqref{KS},
 the properties above yield the manifold of static solutions 
${\cal M}_0:=\{\chi(r/\lambda)\ |\ \lambda>0\}$ 
and suggest a likely scenario of collapse: sliding along ${\cal M}_0$ in the direction of $\lambda \ra 0$.
To analyze the collapse, we  pass to the reference frame collapsing with the solution, by introducing the adaptive blowup variables, 
\begin{equation*}
m(r,t)=u(y,\tau), \mbox{where}\ y=\frac{r}{\lambda}\
\mbox{and}\ \tau=\int_0^t \frac{1}{\lambda^2(s)}\, ds,
\end{equation*}
where $\lambda:[0,T)\rightarrow[0,\infty),\  T>0,$ is  a positive differentiable function (\textit{compression} or \textit{dilatation} parameter), s.t. $\lambda(t)\rightarrow 0$ as $t\uparrow T$. The advantage of moving to blowup variables is that the function $u$ is expected to have bounded derivatives and the blowup time is eliminated from consideration (it is mapped to $\infty$).  Writing \eqref{m-eq} in blowup variables, we find the equation for the rescaled mass function
\begin{equation}\label{u-eq}
\p_\tau u=\Lap{0}_yu+y^{-1}u\p_y u -a y\p_yu,
\end{equation}
where $a:=-\dot{\lambda}\lambda$.

To investigate stability properties of the rescaled stationary solution $\chi(y)$, we decompose solutions $u(y,\tau)$ of equation \eqref{u-eq} into 
 the main term, $\chi(y)$, and  the fluctuation $\phi(y,\tau)$, $u(y,\tau)=\chi(y)+\phi(y,\tau).$ 
Substituting this decomposition into  \eqref{u-eq} gives the equation for the fluctuation $\phi$,
\begin{equation}\label{phi-eq}
 \p_\tau\phi=-L_{a }\phi+\Fab+\Nab(\phi),
\end{equation}
where the forcing and nonlinear terms are $\Fab :=
 -\frac{8a y^2}{(1+ y^2)^2}\ \quad \textrm{and}\ \quad  
 \Nab(\phi):=\frac{1}{y}\phi\p_y\phi ,$
and  the linear operator, $L_{a}$ is given by
 \begin{align}
L_{a} :=-\Lap{4}-\frac{8}{(1+ y^2)^2}+\frac{4}{y(1+ y^2)}\p_y +a y\p_y.
\label{Lab}
\end{align}
An important fact here is that the operator $\La$ is self-adjoint on the space
$L^2(\Rp, \g_{a }(y) y^3 dy)$,
where 
$\g_{a }^{-1/2}(y)= \chi(y) e^{\frac{a}{4} y^2},$ 
 with the inner product 
 $\ip{f}{g}:=\int_0^\infty f(y) g(y)\, \g_{a }(y) y^3 dy.$ 
 One can check the self-adjointness of $L_{a}$ directly or use the unitary map
\begin{equation}\label{gauge-transf}
\phi(y)\ra  \g_{a }^{1/2}(y)\phi(y),
\end{equation} from $ L^2([0,\infty), \g_{a }(y) y^3 dy)$ to
$ L^2([0,\infty), y^3 dy)$, which maps this operator $\La$ into the operator
$\Lab:=\g_{a }^{1/2}L_{a}\g_{a }^{-1/2},$
acting on the space  $  L^2([0,\infty), y^3 dy)$. The latter operator can be explicitly computed to be
\begin{align}\label{Lab-unit}
\Lab :=-\Lap{4}-\frac{8}{(1+ y^2)^2} +\frac{1}{4}a^2 y^2 + \frac{2 a}{1+ y^2}
- 2 a.
\end{align}
  This operator is of the Schr\"odinger type with the real continuous potential tending to $\infty$ as $y \rightarrow \infty$. Therefore, by standard arguments (see e.g. \cite{GS}), it is self-adjoint and its spectrum is purely discrete.  Hence  $\La$ is self-adjoint on the space $ L^2([0,\infty), \g_{a }(y) y^3 dy)$ and has purely discrete spectrum as well. 
 Going through with our analysis shows that $a(\tau)\ra 0$ as $\tau\ra \infty$, which actually complicates the problem and which tells us that the collapse is slower than parabolic one, $\lam(t)= \sqrt{a_0(T-t)}$, for which $a(\tau)=-\lam(t)\dot \lam(t)$ is a constant (say, $a_0$).

Now, it is clear that the stability of the 
the profile $\chi(y)$ is determined largely by the spectrum of the operator $\Lab$. If  the operator $\Lab$ has strictly positive spectrum, one expect the fluctuations $\phi$ will die out as $\tau\ra \infty$ and consequently the solution of \eqref{u-eq} will tend to $\chi(y)$, while  the solution of \eqref{m-eq} will approach $\chi(r/\lam(t))$. On the other hand, if  the operator $\Lab$ has negative eigenvalues then one expects instability. The latter though is 
always the case, since the equations have a negative scaling mode (for a fixed parabolic scaling it is connected to possible variation of the blowup time.) \footnote{A similar analysis applies also in the subcritical case $M<8\pi$ where the solution converges to a self-similar one as $\tau\ra\infty$,  which vanishes as $t\ra\infty$. In this case  the operator $\Lab$ has strictly positive spectrum.}

 If the number of negative eigenvalues is finite, say $k$, then one either goes to an invariant manifold theory and constructs the central-unstable manifold or, uses the (related) modulation theory and embeds $\chi(r/\lam)$ into a $k-$parameter family of almost solutions, say $\chi_{p}(r/\lam)$, where $p$ stands for the $k-1$ parameters (with $\lambda$, or $a$, counted as the first parameter), chosen so that
 \DETAILS{in the decomposition
 \begin{equation} \label{parametriz-p}
u(y,\tau)=\chi_{p(\tau)}(y)+\phi(y,\tau),
\end{equation}
the fluctuation $\phi(y,\tau)$ is (approximately) orthogonal to the negative spectrum eigenfunctions of $\Lab$. If this holds, 
 then the stability is restored and the solution to \eqref{u-eq}  approaches this family as $\tau\ra \infty$. The latter is a big if and sometimes is the hardest part of the analysis. This is where the understanding  the negative spectrum eigenfunctions of $\Lab$ 
 helps, as these eigenfunctions span the tangent space for the deformation (or almost center) manifold $\cM:=\{\chi_{p}(r/\lambda)\ |\ \lambda>0,\ p\}$ and allow one to construct it.}
  the tangent space of the deformation (or almost center-unstable) manifold $\cM:=\{\chi_{p}(r/\lambda)\ |\ \lambda>0,\ p\}$ at $\chi_{p}(y)$ is equal approximately to the eigenspace of negative and (almost) zero spectrum of $\Lab$. Then we can choose the parameters  $p=p(\tau)$ and $a=a(\tau)$ (or  $\lambda =\lambda(\tau)$), so that the solution $u(y,\tau)$ can be decomposed as
 \begin{equation} \label{parametriz-p}
u(y,\tau)=\chi_{p(\tau)}(y)+\phi(y,\tau),
\end{equation}
with the fluctuation $\phi(y,\tau)$ orthogonal to  the tangent space of  $\cM$ at $\chi_{p}(y)$, 
$\lan\p_p\chi_{p(\tau)}(\cdot), \phi(\cdot,\tau)\ran=0,$ 
 and therefore (approximately) orthogonal to the negative and almost zero
spectrum eigenfunctions of $\Lab$. If we find such a deformation, $\chi_{p(\tau)}(y)$, 
then the stability is restored and the solution to \eqref{u-eq}  approaches this family as $\tau\ra \infty$. The latter is a big if and 
this is where the understanding  the negative and almost zero spectrum eigenfunctions of $\Lab$ 
helps. 
(One should keep in in mind that the deformation of $\chi(y)$ will change 
the linearized operator $\Lab$ 
and the gauge transformation \eqref{gauge-transf}, 
but both can be easily handled.)

Theorem \ref{thm:main} below implies that the operator  $\Lab$ of the Keller-Segel system has one negative 
(corresponding to the scaling mode mentioned above) and one near zero 
eigenvalue, while the third eigenvalue, $2a+\frac{2a}{\ln\frac{1}{a}}+\O{a\ln^{-2}\frac{1}{a}}$, is positive, but vanishing as $a\ra 0$. (
 It also isolates the correct perturbation (adiabatic) parameter - $\frac{1}{\ln\frac{1}{a}}$.) Hence we have to construct a one-parameter deformation of $\chi(y)$ (remember that $\lambda$, or $a$, is counted as a parameter). For technical reasons it is convenient to use a two-parameter family,  $\chi_{bc}(y)$, with an extra relation between the parameters $a,\ b$ and $c$.
 In \cite{DLS} we \textit{choose the family}
\begin{equation}\label{chibc}
 \chi_{bc}(y):=\frac{4b y^2}{c+ y^2},
\end{equation}
with $b>1$ and both parameters $b$ and $c$ are close to $1$. Note that this family evolves on a different spatial scale than $\phi(y, \tau)$ in \eqref{parametriz-p}, as it can rewritten as $\chi_{bc}(y)=\chi_{\frac{b}{c}, 1}(\frac{y}{\sqrt{c}})$.
The contraction law is obtained by using  the orthogonality condition, $\lan\p_{bc}\chi_{bc}, \phi\ran=0$. The latter is equivalent to two conditions,
\begin{equation} \label{orthog}
\p_\tau\lan\p_{bc}\chi_{b(\tau)c(\tau)}(\cdot), \phi(\cdot,\tau)\ran=0
\end{equation}
and $\lan\p_{bc}\chi_{b(\tau)c(\tau)}(\cdot), \phi(\cdot,\tau)\ran|_{t=0}=0$. We evaluate \eqref{orthog} by using  the evolution equation, $\p_\tau\phi=-L_{a bc}\phi+\Fabc+\Nab(\phi)$, for $\phi$, similar to \eqref{phi-eq}, which follows by plugging the decomposition \eqref{parametriz-p} into \eqref{u-eq}, and the explicit expressions,
\begin{align}\label{zbc}
 \zeta_{bc1}(y):=\frac{1}{4}\p_b\chi_{bc} (y)=\frac{y^{2}}{c+ y^{2}},\ \quad \zeta_{bc2}(y):=\frac{1}{4b}\p_c\chi_{bc} (y)=\frac{y^{2}}{(c+ y^{2})^2},
\end{align}
for the tangent vectors to the manifold $\cM$. (The scaling mode $\zeta_{bc0}(y):=\frac{1}{8 bc}y\p_y\chi_{bc} (y)=\frac{ y^{2}}{(c+y^{2})^2}$ a multiple of $\zeta_{bc2}(y)$ which confirms that one of the parameters is superfluous.)  This gives ordinary differential equations for $a$, $b$ and $c$ with higher order terms depending on 
$\phi$:
%
\begin{align}\label{abc-eqs}
\begin{cases}
&c_\tau +
S(\phi, a, b, c) a_\tau=2 a - \frac{4d}{\ln(\frac{1}{a})}+ R(\phi, a, b, c),\\
 & \frac{d_\tau}{a}-
S(\phi, a, b, c) a_\tau =-\frac{2d}{\ln(\frac{1}{a})}+  R(\phi, a, b, c), 
\end{cases}\end{align}
where $d:=b-1$, $|S(\phi, a, b, c)|\lesssim \frac{\|\phi\|_{L^2}}{a^{d+1}\ln(\frac{1}{a})}$ and
$|R(\phi, a, b, c)|\lesssim \frac{a}{\ln^2(\frac{1}{a})} \frac{1}{\ln(\frac{1}{a})}[d\|\phi\|_{L^2}+\|(1+y^2)^{-1}\phi\|_{L^2}^2].$ 
%
These  higher order terms are controlled by using a differential inequality for the Lyapunov functional $\phi\mapsto\|\phi\|_{L^2}^2$ and the inequality $\lan\phi, \Labc\phi\ran\ge 2a\|\phi\|^2$, which follows from our result below. \footnote{Presently, without taking into account the nonlinearity in the equation $\p_\tau\phi=-L_{a bc}\phi+\Fabc+\Nab(\phi)$
(\cite{DLS}).} 
 (One can also use higher order Lyapunov functionals like $\lan\phi, \Labc^k\phi\ran,\ k\ge 1$. The fact that the positive eigenvalues of $L_{a bc}$ vanish makes estimating $\phi$ a delicate matter.)
 Finally, we choose a relation between  $a$, $b$ and $c$ so as to eliminate large terms in the corresponding vector fields, namely, $d=\frac{1}{2}a\ln(\frac{1}{a})$,
This leads, in the leading order, to the differential equation
\begin{align}\label{orth-eqns-lead7}
 a_\tau = -  \frac{2a^2}{\ln(\frac{1}{a})},
\end{align}
whose solutions, in the leading order, are $\frac{1}{a}\left (\ln{\frac{1}{a}}+O(1) \right )=2\tau$ 
 which results in
$ \ln{\frac{1}{a(\tau)}}=\ln{2\tau}-\ln{\ln{2\tau}}+\frac{\ln{\ln{2\tau}}}{\ln{2\tau}}+O\left (\frac{1}{\ln{2\tau}}\right )$. 
Recalling that $\lambda (t)\dot\lambda (t)=-a(\tau)$ and using that $\lambda (t)\dot\lambda (t)=\lambda (\tau(t))^{-1}\p_\tau\lambda(\tau(t))$, we obtain
that $-\ln{\lambda (\tau)}= \int^\tau a(\tau)d\tau$ giving
\begin{align}\label{lnlambda}
-\ln \lambda (t)=
\frac{(\ln{2\tau})^2}{4}-\frac{(\ln{2\tau})\ln{\ln{2\tau}}}{2}+O(\ln{2\tau}),
\end{align}
while $\tau$ is related to $t$ by
\begin{align}\label{Ttdef}
(T-t)=  \int^\infty_\tau \lambda(\tau')^2d\tau'= (1/2)\left  [(\ln{2\tau})^{-1}+O((\ln{2\tau})^{-2}\ln{\ln{2\tau}})\right ] \nonumber  \\
\times \exp{\left[ -\frac{(\ln{2\tau})^2}{2}+(\ln{2\tau})(\ln{\ln{2\tau}})+O(\ln{2\tau})\right ]},
\end{align}
where the integral was computed asymptotically in a limit $\tau \gg 1$ and $T$ is a constant of integration determined by initial conditions.
Solving the equations (\ref{lnlambda}) and (\ref{Ttdef})  together  for $t\to T$ yields the law
\begin{align}\lambda(t)=(T-t)^\frac{1}{2}
e^{-|\frac{1}{2}\ln(T-t)|^\frac{1}{2}}(c_1+o(1)),
\end{align}
%
which  coincides in the leading order (up to the constant $c_1$) with the one obtained in \cite{Ve1, Lu,DyachenkoLushnikovVladimirovaAIPConfProc2011}.
The constant $c_1$ can be obtained only if we consider a next order correction beyond the accuracy of the equation (\ref{orth-eqns-lead7}) (see \cite{Lu,DyachenkoLushnikovVladimirovaAIPConfProc2011})
which is outside the scope of
this paper.

The above arguments show that the  spectral analysis of the linearized equation on the the collapse or blowup profile 
is the key step in describing  critical collapse or blowup laws for nonlinear evolution equations. (This also applies to stability analysis of stationary and traveling wave solutions.)
 Typically, this is a rather subtle affair with very few general techniques available. In this paper we develop such a technique for 
differential operators of the form \eqref{Lab-unit}, 
\begin{equation}
{\cal L}:=
-\Lap{4}-\frac{8}{(1+y^2)^2}+\frac{1}{4}a^2 y^2+\Wa(y),
\label{OperatorL}
\end{equation}
 defined on the space $L^2([0,\infty), y^3 dy)$, which is the subspace of radially symmetric functions in $L^2(\R^4)$.  Here, as indicated 
above, the parameter $a$ is assumed small and positive and the potential $\Wa$, 
to satisfy the bound
\begin{equation}
 0 \le  \Wa(x)\le\frac{ C a}{1+y^2},
\label{Wa}
\end{equation}
where $C$ is a positive constant. We assume that $W_a(x)$ is positive in order to fix the bottom of the spectrum: now ${\cal L} \ge 0$ (see below). 
Our main result is the derivation of an approximate equation for the low-lying eigenvalues of ${\cal L}$, which enter into the stability analysis mentioned above.  Let $\Psi$ be the digamma function, defined by $\Psi(s)=\frac{d}{ds}\ln \Gamma(s)$,
where $\Gamma(s)$ is the gamma function.  Define
\begin{equation*}
0 \le \mu:=2\int_0^\infty \frac{\Wa(y)}{(1+y^2)^2}\, y^3 dy\le C a.
\end{equation*}
and $K:=\ln 2-1-2\gamma$, where $\gamma =-\Psi(1)=0.577216\ldots$ is the Euler-Mascheroni
constant. We have
\begin{thm} \label{thm:main}
The operator ${\cal L}$ is self-adjoint on $L^2([0,\infty), y^3 dy)$, positive and its spectrum is discrete.  For $a$ small, the eigenvalues of ${\cal L}$ in the interval $[0, Ca]$, for any given $C>0$, satisfy the equation
\begin{equation}
\frac{\lambda}{\mu+a}\lsb \ln\frac{1}{a}-\Psi\lb 1-\frac{\lambda}{2a} \rb+K \rsb=1+\O{a^{1/2}\ln\frac{1}{a}} .
\label{eqn:eigenvalue}
\end{equation}
\end{thm}
As $a\rightarrow 0$, solutions of
\begin{equation} \label{eveqn}
\frac{\lambda}{\mu+a}\lsb \ln\frac{1}{a}-\Psi\lb 1-\frac{\lambda}{2a} \rb+K \rsb=1
\end{equation}
converge to the eigenvalues of ${\cal L}$.  We solve this equation approximately in Section \ref{solving} to obtain
\begin{equation}
 \lambda_{n}=\left\{\begin{array}{ll}\frac{\mu+a}{\ln\frac{1}{a}+K+\gamma}+\O{a\ln^{-3}\frac{1}{a}} &n=0\\
2n a+\frac{2a}{\ln\frac{1}{a}+K+\gamma-H_{n-1}-\frac{\mu+a}{2a n}}+\O{a\ln^{-3}\frac{1}{a}} & n\ge 1,
\end{array}\right.
\label{eqn:explicit_eigvalue_estimates}
\end{equation}
where $H_{n}:=\sum_{k=1}^n 1/k$.  These approximations to the eigenvalues, especially the one obtained by solving numerically \eqref{eveqn}, match remarkably well with the numerical computation of the spectrum of ${\cal L}$, the results of which are
given in Figure \ref{ev.plot} below. The fact that the numerical solution to  \eqref{eveqn} gives much better approximation to the true eigenvalues is not surprising:
the approximation  \eqref{eqn:explicit_eigvalue_estimates} has the logarithmic error while the equation  \eqref{eveqn} is obtained with the power accuracy.

Figure \ref{ev.plot} compares the eigenvalue approximations obtained using  \eqref{eveqn} and
\eqref{eqn:explicit_eigvalue_estimates}  against the numerical computation of the first three eigenvalues of ${\cal L}$. We have taken $\Wa=2a/(1+y^2)$ (this gives $\mu=a$). Numerical procedure is explained in Section \ref{numerics}.    The computations confirm the spectral picture we have obtained analytically with the high precision numerical computations.
\begin{figure}[ht]
\begin{center}
\includegraphics[width=5in] {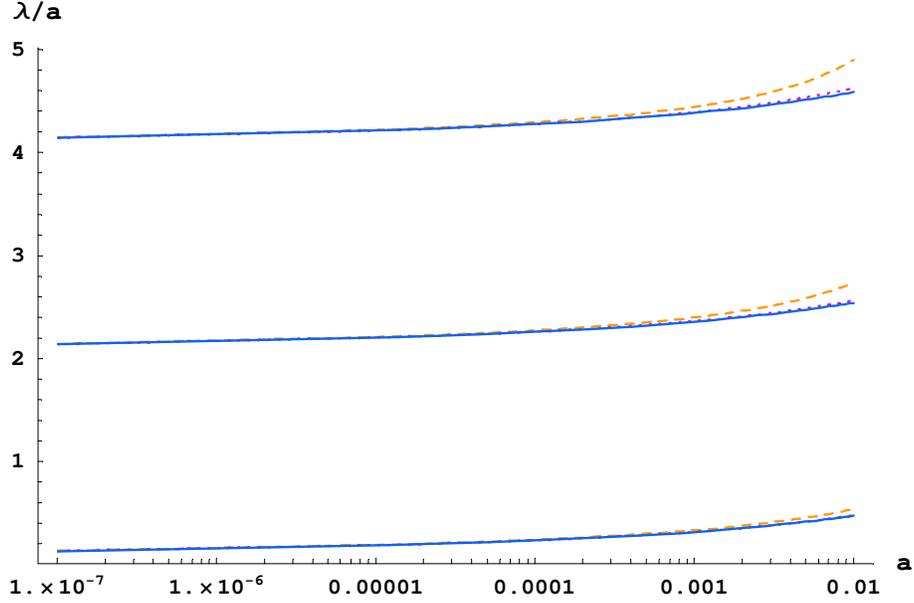}
\caption{The quantity $\lambda_i/a$ is plotted against $a$ for the first three eigenvalues obtained in three different ways.
The solid lines are numerically computed, the dashed lines are obtained using the expressions in \eqref{eqn:explicit_eigvalue_estimates}, and the dotted line (coinciding with the numerical
computation for the first eigenvalue) is obtained by numerically solving the eigenvalue equation \eqref{eveqn}. 
The eigenvalues plots obtained by solving the full equation and the equation  \eqref{eveqn} are very close.}
\label{ev.plot}
\end{center}
\end{figure}

We analyze the spectrum of ${\cal L}$ in the interval $[0,C a]$ for any fixed $C$ independent of $a$.  This is sufficient for the stability analysis for the problem described above.  However, we believe that our results are valid in a larger interval.

 We have already discussed the self-adjointness of ${\cal L}$ on $L^2([0,\infty), y^3 dy)$ and the discreteness of its spectrum. 
The scaling symmetry of \eqref{KS} implies that the function $\eta_1(y):=1/(1+y^2)$ is a null vector of the operator
\begin{equation*}
\cL_0:= 
-\Lap{4} -\frac{8}{(1+y^2)^2}.
\end{equation*}
By the Perron-Frobenius argument this implies that the above operator is non-negative and has the non-degenerate eigenvalue at $0$. The first fact implies that  ${\cal L}\ge 0$.

As was mentioned above, Theorem \ref{thm:main} is proven by a rigorous version of the method of matched asymptotics (see \cite{BeOr1999}).Though this method is standard; other instances of its use in spectral
problems can be found in \cite{Ba1996, BoNa1998,  Br1968, DoGaKa1998,LushnikovAtm1998rus}, we, however, believe that our extension of this method is novel and robust and can be
used a large variety of linear differential operators arising in the linearization of nonlinear equations and hopefully can be extended to nonlinear ones as well (in this case perturbation series below
should be replaced by fixed point iterations).

Indeed,
like arguments outlined above, our approach is fairly general. It requires essentially only the properties listed above: the scaling symmetry and existence of a stationary solution. (In case of asymptotic motion of solitons, the scaling symmetry is replaced by translational, or more generally Galilean or Poicar\'e, invariance.)
Indeed, the potential term $-\frac{8}{(1+ y^2)^2}$ comes from linearlizing the nonlinear part of the equation on the stationary solution $\chi(y)$ ($\frac{8}{(1+ y^2)^2}=\frac{1}{y}\p_y\chi(y)$), while  the confining potential $\frac{1}{4}a^2 y^2$ comes from the term (vector-field), $a y\p_y$, generated by the time-dependent rescaling, and it occurs in all such problems. We do not use the explicit form of $-\frac{8}{(1+ y^2)^2}$ (besides of its asymptotics at $y\ra \infty$ and at $y\ra 0$), but the fact that, since  the stationary solution, $\chi(y)$, breaks the scaling symmetry, it leads to the zero mode  $\eta_1(y):= y\p_y\chi(y)$ of the original linearization,
$\cL_0:=\g_{0 }^{1/2}L_{0}\g_{0}^{-1/2}=-\Lap{4}-\frac{8}{(1+ y^2)^2}$ (after the transformation $\xi\ra \g_{0 }^{1/2}\xi$).

\DETAILS{The operator ${\cal L}$ appears in the analysis of the well known Keller-Segel equation (also called Patlak-Keller-Segel equation) \cite{BrCoKaScVe1999,Lu,DeLuSi2010} introduced in \cite{Patlak1953,KeSe1970} as a model for chemotactic behaviour of biological organisms.  After simplifying assumptions, this system becomes
\begin{align}
\begin{split}
 \p_t \rho&=\Delta\rho-\nabla\cdot(\rho\nabla c),\\
0&=\Delta c+\rho,
\end{split}
\label{KS}
\end{align}
with $\rho$ and $c$ satisfying the no-flux Neumann boundary conditions. It is known that solutions to this equation can develop,  in finite time, infinite mass at a point in space.   Scaling symmetry,
\begin{align*}
\rho(x,t)\rightarrow\frac{1}{\lambda^2}\rho\lb\frac{1}{\lambda}x,\frac{1}{\lambda^2}t\rb\
\mbox{and}\ c(x,t)\rightarrow c\lb\frac{1}{\lambda} x,\frac{1}{\lambda^2} t\rb,
\end{align*}
and existence of the static solution,
$\rho_{st}(x):=\frac{8}{(1+|x|^2)^2}\ \mbox{and}\ c_{st}(x):=-2\ln(1+|x|^2),$
suggest that this occurs along the manifold ${\cal M}:=\{\ (\lambda^{-2}\rho_{st}(x/\lambda), c_{st}(x/\lambda))\ |\ \lambda>0\}$ generated by scaled static solutions.  The linearization of \eqref{KS} about ${\cal M}$ and the parametrization
\begin{equation*}
\rho(x,t)=\lambda^{-2}(t)\rho_{st}(x/\lambda(t))+f(x/\lambda(t))\eta(x/\lambda(t),t)
\end{equation*}
of $\rho(x, t)$, with an appropriate $f(y)$, produces the operator ${\cal L}$.  Knowledge of the spectrum of ${\cal L}$ determines the linear stability of ${\cal M}$.  For more details we refer to \cite{DeLuSi2010, Lu}.}

Finally, we mention the major limitation of our set-up - we deal with radially symmetric solutions. Since the only stationary solution is  radially symmetric (unfortunately, in contrast to biological observations), the linearized operator for the full equation commutes
with rotations and therefore can be decomposed in the direct sum of radial operators. Hopefully, our analysis can be extended to each component of this sum.

This paper is organized as follows.  In Section \ref{sec:inner}, using perturbation theory, we solve the eigenvalue problem 
\begin{equation}
{\cal L}\phi_{\lambda}=\lambda\phi_\lambda
\label{eqn:EigenValueProblem}
\end{equation}
in the inner region and then proceed to find the leading order expression.  We also use perturbation theory in Section \ref{sec:outer} to solve \eqref{eqn:EigenValueProblem} in the outer region and 
find the leading order behaviour of this solution. 
In Section \ref{matching} we match the inner and outer solutions and in Section \ref{solving} we solve the equation \eqref{eveqn} to obtain the solutions \eqref{eqn:explicit_eigvalue_estimates}.  Finally, in Section \ref{numerics} we briefly discuss our numerical procedure. In Appendix A we give explicit derivations of some of the expressions of Section \ref{sec:outer}, which were obtained with reference to the theory of special functions.

In what follows, we use the notation $f\lesssim g$ for $f, g \ge 0$, if there exists a positive constant $C$ such that $f\le C g$ holds.  If the inequality $|f|\le C|g|$ holds then we write $f=\O{g}$.  We also write $f\ll g$ or $f=\smallO{}{g}$ if $f/g\rightarrow 0$ as $a$ or $y$ approach some limit (always specified) and $f\sim g$ if the quotient converges to 1.

\section{Solutions in the Inner Region}
\label{sec:inner}
In what follows, $\lambda >0$ is a spectral parameter and $0 < a\ll 1$. To simplify the expression we assume in what follows that $\lambda \le Ca$ for some $C>0$. Below we take $\Rin=\epsin/a^\frac{1}{2}$ with $a\ll \epsin\ll 
 \sqrt{\frac{a}{\lambda}\frac{1}{\ln\frac{1}{a}}}$.
 The main result of this section is the following
\begin{prop}
\label{prop:innersol}
If $a$ is small enough, then the solution to the eigenvalue problem \eqref{eqn:EigenValueProblem} in the inner region $[0,\Rin]$ is unique, modulo an overall constant factor. For $y\in[\Rout,\Rin]$, $\Rout\gg 1$ as $a\rightarrow 0$, this solution is given by
\begin{equation}
\phi_\lambda^{in}=\frac{1}{y^2}-\frac{1}{4}\lambda\ln y^2+\frac{1}{4}(2\lambda+\mu)
+{\cal R}_i,
\label{phiinner}
\end{equation}
where
\begin{equation}
{\cal R}_i=\O{\lambda\frac{\ln^2 y}{y^2}+\frac{1}{y^4} +a^2 y^2+ \epsin^4+ \lsb\frac{\lambda}{a}\epsin^2\ln\frac{1}{a}\rsb^2\frac{1}{y^2}}
\label{reminner}
\end{equation}
and the $\O{f}$ signifies the bound $|\O{f}| \le C f$ with a uniform constant $C$. 
\end{prop}
\begin{proof} For $\Rin>0$, let $\Lin{\cdot}$ be the norm defined on measurable functions $f:[0,\Rin]\rightarrow \R$ by
\begin{equation*}
 \Lin{f}:=\sup_{y\le \Rin}\left| (1+y^2) f\right| 
\end{equation*}
and let $\Xin:=\left\{ f:[0,\Rin]\rightarrow \R\ |\ \Lin{f}<\infty \right\}$ be the associated Banach space.

Equation \eqref{eqn:EigenValueProblem} can be written as $(\Lopin+\Vin)\phi_\lambda=0$, where $\Lopin$ and $\Vin$ are defined
as
\begin{equation*}
\Lopin:=-\frac{d^2}{dy^2}-\frac{3}{y}\frac{d}{dy}-\frac{8}{(1+y^2)^2}\ \mbox{and}\ \Vin:=-\lambda+\frac{1}{4}a^2 y^2+\Wa(y).
\end{equation*}
The operator $\Lopin$ is self-adjoint on $L^2([0,\infty), y^3 dy)$ with essential spectrum $\sigma_{ess}(\Lopin)=[0, \infty)$.  It is straightforward to check that the function
\begin{equation*}
 \eta_1(y):=\frac{1}{1+y^2}
\end{equation*}
is a solution to the zero mode equation $\Lopin\eta=0$ (as mentioned in the introduction, this is due to the scaling symmetry of \eqref{KS}).  Moreover, positivity of $\eta_1$ and the Perron-Frobenius theorem imply that $0=\inf \sigma(\Lopin)$, and hence $\Lopin$ is a positive operator, but 0 is not an eigenvalue.  The function $\eta_1$ is not an eigenvector of $\Lopin$ since it does not lie in $L^2([0,\infty), y^3 dy)$; we call it a resonance eigenvector.  Using that the Wronskian of $\Lopin$ is $1/y^3$ , we find a second solution of $\Lopin\eta=0$:
\begin{equation*}
 \eta_2:=\eta_1\lb \frac{1}{2} y^2+\ln y^2-\frac{1}{2} y^{-2} \rb.
\end{equation*}
This vector is independent of $\eta_1$, also does not lie in $L^2([0,\infty), y^3 dy)$ and has a singularity at $y=0$. Note that $\eta \in \Xin$.

Using variation of parameters we find that the general solution to $\Lopin\xi=f$ is $\xi(y)=(\Lopin^{-1} f)(y)+C_1\eta_1(y)+C_2\eta_2(y)$, where $C_1$ and $C_2$ are arbitrary constants and
\begin{equation}
(\Lopin^{-1} f)(y):=\eta_1(y)\int_0^y \eta_2(x)f(x) \, x^3 dx-\eta_2(y)\int_0^y \eta_1(x) f(x)\, x^3 dx.
\label{Lopininv}
\end{equation}
\begin{lemma}
\label{prop:inner}
Say that $a\ll \epsin\ll 1$. 
For $a$ small enough,
Equation \eqref{eqn:EigenValueProblem} has a unique, modulo an overall constant factor, solution on $[0,\Rin]$ of the form $\phi_\lambda^{in}=\eta_1+\xi$ with $\xi\in\Xin$ and
\begin{equation} \label{inNeumann}
 \xi=\sum_{n=1}^\infty (-\Lopin^{-1} \Vin)^n\eta_1.
\end{equation}
The convergence is in $\Xin$ and $\|\Lopin^{-1} \Vin\|_{\Xin\rightarrow\Xin}\lesssim (\frac{\lambda}{a}+1)\epsin^2\ln\frac{1}{a}$.
\end{lemma}
\begin{remark} In fact, we can show convergence of the series in an appropriate norm without the condition on the parameter $a$ and the range of $y$.
\end{remark}
\begin{proof}
Recall that ${\cal L}=\Lopin+\Vin$.  Substituting $\phi_\lambda^{in}=\eta_1+\xi$ into \eqref{eqn:EigenValueProblem} and using $\Lopin\eta_1=0$ and the general solution of $\Lopin\xi=f$ found above, we obtain that $\xi=-\Lopin^{-1}(\Vin\xi+\Vin\eta)+C_1\eta_1+C_2\eta_2$.  The term $C_1\eta_1$ leads to change of an overall factor multiplying $\phi_\lambda^{in}$ and therefore it can be dropped.  Next, if $\xi\in\Xin$, then so is $\Lopin^{-1}\Vin\xi$ (see below).  Since $\eta_2\not\in \Xin$, we must, therefore, take $C_2=0$, otherwise we would have a contradiction.  Finally, we rearrange the resulting equation to obtain that
\begin{equation} \label{ininteqn}
 (I+\Lopin^{-1} \Vin)\xi=-\Lopin^{-1} \Vin\eta_1.
\end{equation}
As \eqref{Lopininv} shows, \eqref{ininteqn} is a Volterra equation and therefore the operator on the r.h.s. has an inverse and this inverse can be expanded in the standard perturbation series. We can invert the operator on the left hand side on the space $\Xin$, provided $\|\Lopin^{-1} \Vin\|_{\Xin\rightarrow\Xin}<1$.  To show the latter property, we compute that
\begin{equation}
\begin{split}
\Lin{\Lopin^{-1} \Vin f}& \le \left\{\sup_{y\le \Rin} |\rho(y)\eta_1(y)|\int_0^y \frac{|\eta_2(z) \Vin(z)|}{\rho(z)} \,z^3 dz \right.\\
&\left.+\sup_{y\le \Rin}|\rho(y)\eta_2(y)|\int_0^y \frac{|\eta_1(z) \Vin(z)|}{\rho(z)} \,z^3 dz\right\}\Lin{f},
\end{split}
\label{Linnerbound}
\end{equation}
where $\rho (y)=(1+y^2)$ 
and $f\in\Xin$.  Substituting the expressions for $\rho$, $\eta_1$, $\eta_2$ and $\Vin$ into the first term and using \eqref{Wa} gives that
\begin{align}
 |\rho\eta_1|\int_0^y \frac{|\eta_2 \Vin|}{\rho}(z)\, z^3 dz&=\int_0^y \frac{1}{(1+z^2)^2}\left| \frac{1}{2}z^2+\ln z^2-\frac{1}{2}\frac{1}{z^2}\right|\left| -\lambda+\frac{a^2}{4} z^2+\Wa(z) \right|\, z^3 dz\nonumber\\
&\lesssim\int_0^y\lb \lambda+a^2 z^2+\frac{a}{1+z^2}\rb z\, dz.\label{InnerPart1Bound}
\end{align}
This gives that
$|\rho\eta_1|\int_0^y \frac{|\eta_2 \Vin|}{\rho}(z)\, z^3 dz\lesssim a^2 y^4+a\ln(1+ y^2)+\lambda y^2$,
and hence
\begin{equation}
 \sup_{[0,\Rin]}|\rho\eta_1|\int_0^y \frac{|\eta_2 \Vin|}{\rho}(z)\, z^3 dz\lesssim a^2 \Rin^4+a\ln\Rin^2+\lambda\Rin^2.
\label{part1boundinner}
\end{equation}
Similarly, we compute that
\begin{equation*}
|\rho\eta_2|\int_0^y \frac{|\eta_1 \Vin|}{\rho}(z)\, z^3 dz\lesssim 
\lb y^2+\frac{1}{y^2} \rb\lb ay^4+\lambda y^4 (1+\ln(1+y^2))+ a^2y^6\rb(1+y^2)^{-2} 
\end{equation*}
and hence,
\begin{equation}
\sup_{[0,\Rin]}|\rho\eta_2|\int_0^y \frac{|\eta_1 \Vin|}{\rho}(z)\, z^3 dz\lesssim a^2 \Rin^4+\lambda\Rin^2\ln\Rin+(a+\lambda)\Rin^2. 
\label{part2boundinner}
\end{equation}
Substituting the definition  $\Rin:= \epsin/ a^{1/2}$ into \eqref{part1boundinner} and \eqref{part2boundinner}, using $1\gg\epsin\gg a$ to simplify $\ln\frac{\epsin^2}{a}$ to $\ln\frac{1}{a}$ and then using  
the results in \eqref{Linnerbound} gives that $\|\Lopin^{-1}\Vin\|_{\Xin\rightarrow\Xin}\lesssim \epsin^2\frac{\lambda}{a} \ln\frac{1}{a}$. 
and hence 
$a$ can be taken small enough so that $\|\Lopin^{-1} \Vin\|_{\Xin\rightarrow\Xin}<1$. Now inverting the operator on the l.h.s. of \eqref{ininteqn} and expanding the inverse into the Neumann series completes the proof.
\end{proof}

The expression \eqref{phiinner} for $\phi_\lambda^{in}$
is obtained as follows.  Due to Lemma \ref{prop:inner} and since $\Lin{\eta_1}=\O{1}$, if $a$ is small enough (and $\epsin\ll 1$), 
 then
\begin{equation}
 \phi_\lambda=\eta_1-\Lopin^{-1} \Vin \eta_1+\OX{\Xin}{\epsin^4+\lsb\frac{\lambda}{a}\epsin^2\ln\frac{1}{a}\rsb^2},
\label{ex:phiInnerV0}
\end{equation}
where the $\OX{\Xin}{\epsilon}$ stands for a function bounded as $\|\OX{\Xin}{\epsilon}\|_{\Xin} \lesssim \epsilon$. Since $|f| \le \frac{1 }{1+y^2} \|f\|_i$, 
 the latter is the same as $|\OX{\Xin}{\epsilon}| \lesssim \epsilon (1+y^2)^{-1}$. Substituting the large $y$ expansions $\eta_1=y^{-2}+\O{y^{-4}}$ and $\eta_2=1/2+\O{\ln(y)/y^2}$ into the expression for $\Lopin^{-1} \Vin\eta_1$ (see \eqref{Lopininv}) and keeping only leading terms needed in the region $[\Rout,\Rin]$ gives the expression
\begin{equation*}
\Lopin^{-1} \Vin\eta_1=-\frac{2\lambda +\mu}{4}+\frac{\lambda}{4}\ln y^2+\O{\lambda\frac{\ln^2 y}{y^2}+(a+|\mu|)\frac{\ln y}{y^2}+a^2 y^2}.
\end{equation*}
Substituting this expression into \eqref{ex:phiInnerV0}, with $\eta_1$ replaced with its large $y$ expansion and using that $$|\OX{\Xin}{\epsin^4+\lsb\frac{\lambda}{a}\epsin^2\ln\frac{1}{a}\rsb^2}| \lesssim \lb \epsin^4+\lsb\frac{\lambda}{a}\epsin^2\ln\frac{1}{a}\rsb^2\rb (1+y^2)^{-1},$$
gives \eqref{phiinner}. \end{proof}

\section{Solutions in the Outer Region}
\label{sec:outer}
 In the following discussion we take $\Rout:=\epsout/a^\frac{1}{2}$ with $\frac{a^\frac{1}{2}}{\ln^\frac{1}{4}\frac{1}{a}}\ll\epsout\ll 1$ as $a\rightarrow 0$.  
 The main result of this section is the following
 \begin{prop}
\label{prop:outersol}
 On $[\Rout,\infty)$  \eqref{eqn:EigenValueProblem} has a unique solution $\phi_\lambda^{out}$, which, for $y\in[\Rout,\Rin]$, takes the form
\begin{equation}
 \phi_\lambda^{out}= \ln y^2-\ln\frac{1}{a}-\ln 2-1+2\gamma+\Psi\lb 1- \frac{\lambda}{2a}\rb -\frac{4}{\lambda}\frac{1}{y^2}+\frac{a}{\lambda} +{\cal R}_o,
\label{eqn:phiouter}
\end{equation}
where 
\begin{equation}
 {\cal R}_o=\O{a y^2|\ln (a y^2)|+\frac{|\ln(aR_o^2)|}{R_o^2} (a y^2 )^{\frac{\lambda}{2a}-1}}.  
\label{eqn:remouter}
\end{equation}
\end{prop}
\begin{proof} For $a$, $\Rout$ and 
on measurable functions $f:[\Rout,\infty)\rightarrow\R$, we define the norm
\begin{equation*}
 \Lout{f}:=\sup_{y\ge \Rout}\left| 
 a y^2 (a y^2 +1)^{-\frac{\lambda}{2a}}e^{\frac{a}{4} y^2}f \right|.
\end{equation*}
Let $\Xout$ be the corresponding Banach space of functions defined on $[\Rout,0)$ with finite norm $\Lout{f}$. 

We write \eqref{eqn:EigenValueProblem}   as $(\Lopout+\Vout)\phi_\lambda=0$, where
\begin{equation} \label{V0def1}
 \Lopout:=-\frac{d^2}{dy^2}-\frac{3}{y}\frac{d}{dy}+\Va(a y)-\lambda\ \mbox{and}\ \Vout:=U(y)+\Wa(y).
\end{equation}
(Recall that $\Va(a y)=\frac{1}{4}a^2 y^2$ and $U(y)=-\frac{8}{(1+y^2)^2}$.) In the outer region $y\ge \Rout$, we treat $\Vout$ as a small potential.  The operator $\Lopout$ is self-adjoint on $L^2([0,\infty), y^3 dy)$.

In what follows the relation $f\sim g$ as $y\rightarrow \infty$ means that $f/g$ converges to a constant (which might depend only on $\frac{\lambda}{2a}$) as $y\rightarrow \infty$ and the notation $f\approx g$ as $y\rightarrow 0$ means that $f/g$ converges to $1$ as $y\rightarrow 0$.

We show in Appendix, Proposition \ref{prop:phiibnd}, that  the equation $\Lopout\phi=0$ which is the eigenvalue equation for the spherically symmetric harmonic oscillator in $D=4$,  has two linearly independent solutions, $\phi_0$ and $\phi_1$, satisfying the estimates
\begin{equation} \label{phi0bnd}
|\phi_0| \lesssim \left |\Gamma\left (-\frac{\lambda}{2a}\right )\right |(a y^2)^{-1} (a y^2 +1)^{\frac{\lambda}{2a}}, 
e^{-\frac{a}{4} y^2}
\end{equation}
and
\begin{equation}  \label{phi1bnd}
|\phi_1| \lesssim\frac{1}{\left |\Gamma\left (1-\frac{\lambda}{2a}\right )\right |} (a y^2 +1)^{-\frac{\lambda}{2a}-1}
e^{\frac{a}{4} y^2},
\end{equation}
and having the  Wronskian
 \begin{equation}
W(\phi_0,\phi_1)=-\frac{8}{\lambda y^3}.
\label{eqn:wronskian1}
\end{equation}
Hence $\phi_0 \in \Xout $. Using variation of parameters and the Wronskian (\ref{eqn:wronskian1}),
 we find the general solution to $\Lopout\xi=f$ as $\xi=\Lopout^{-1}f+C_1\phi_0+C_2\phi_1$, where $C_1$ and $C_2$ are arbitrary constants and
\begin{equation*}
\Lopout^{-1} f:= -\frac{\lambda}{8}
(\phi_0\int_y^\infty \phi_1 f\, y^3 dy-\phi_1\int_y^\infty \phi_0 f\, y^3 dy).
\end{equation*}
\begin{lemma}
\label{prop:outer}
Say that $\frac{\lambda}{\epsout^2}|\ln (\epsout^2)|\ll 1$. 
If the parameter $a$ is small enough, the equation \eqref{eqn:EigenValueProblem} has a unique, modulo and overall constant factor, solution on $[\Rout,\infty)$ of the form $\phi_0+\xi$, with $\xi\in \Xout$ and
\begin{equation} \label{oseries}
 \xi=\sum_{n=1}^\infty (-\Lopout^{-1} \Vout)^n\phi_0.
\end{equation}
The series converges absolutely in $\Xout$ and
\begin{equation} \label{kernelbnd}
\| \Lopout^{-1} \Vout \|_{\Xout\rightarrow\Xout}\lesssim \frac{1}{\Rout^2}|\ln (a\Rout^2)|. 
\end{equation}
\end{lemma}
\begin{proof}
Substituting $\phi_\lambda^{out}=\phi_0+\xi$ into $(\Lopout+\Vout)\phi=0$ and using that
$\Lopout\phi_0=0$, we obtain $\Lopout \xi+\Vout \xi=-V_o \phi_o$. Now, using the form of the general solution to $\Lopout\xi=f$ found above and that $\phi_1\not\in\Xout$ gives that $(I+\Lopout^{-1} \Vout)\xi=-\Lopout^{-1} \Vout\phi_0$.  We note that the choice of the constants $C_1=0$ and $C_2=0$ follows from similar arguments as in Proposition \ref{prop:inner}).  The operator on the left hand side can be inverted using the Neumann series if $\|\Lopout^{-1} \Vout\|_{\Xout\rightarrow\Xout}<1$, and hence we estimate $\Lout{\Lopout^{-1}\Vout f}$ for $f\in\Xout$:
\begin{equation}
 \Lout{\Lopout^{-1}\Vout f}\le
\frac{\lambda}{8}\sup_{y\ge \Rout}  \left\{|\rho\phi_0|\int_y^\infty \left| \frac{\phi_1 \Vout}{\rho} \right|\, z^3 dz+ |\rho\phi_1|\int_y^\infty \left| \frac{\phi_0 \Vout}{\rho} \right|\, z^3 dz\right\}\Lout{f},
\label{est:Lout}
\end{equation}
where $\rho (y):=a y^2 (a y^2 +1)^{-\frac{\lambda}{2a}}e^{\frac{a}{4} y^2} 
$.  For $y\ge \Rout$ and $\Rout$ large, 
we have, using \eqref{Wa} and  $a^\frac{1}{2}\Rout\ll 1$, that
\begin{equation*}
 |\Vout(y)|\lesssim \lb a+\frac{1}{\Rout^2} \rb\frac{1}{y^2}\lesssim\frac{1}{\Rout^2}\frac{1}{y^2}.
\end{equation*}
Using this estimate and \eqref{phi0bnd} and \eqref{phi1bnd}  in \eqref{est:Lout} yields
\begin{equation}
\Lout{\Lopout^{-1}\Vout f} \lesssim
\frac{\lambda}{\Rout^2}\sup_{y\ge \Rout}  \left\{\int_y^\infty  \rho_0\rho_1 \, z dz+ \rho_0^{-1}\rho_1e^{\frac{a}{2} y^2}\int_y^\infty \rho_{0}^2\,e^{-\frac{a}{2} z^2} z dz\right\}\Lout{f},
\label{est:Gout}
\end{equation}
where $\rho_0=\rho_0(\frac{ay^2}{2} )$ and $\rho_1=\rho_1(\frac{ay^2}{2} )$ are the prefactors on the r.h.s. of \eqref{phi0bnd} and \eqref{phi1bnd}, respectively. Changing the variables of integration as $t=\frac{ay^2}{2}$ and using that $a\Rout^2=\epsout^2 \ll 1$, we obtain \eqref{kernelbnd}.

%
%
%
%
%
 %
Hence, for $a$ is small enough so that the r.h.s. of \eqref{kernelbnd} $<1$, we have that $\| \Lopout^{-1}\Vout \|_{\Xout\rightarrow\Xout}<1$,  and therefore the series \eqref{oseries} converges absolutely, 
which completes the proof.
\end{proof}

Lemma \ref{prop:outer} shows that in $[\Rout,\infty)$, $\phi_0$ solves \eqref{eqn:EigenValueProblem} in leading order in $a$:
\begin{equation}
 \phi_\lambda^{out}=\phi_0+ \OX{\Xout}{\frac{1}{\Rout^2}|\ln (a\Rout^2)|}. 
\label{philambdalo}
\end{equation}
Next, we show in Appendix that $\phi_0$ has the following asymptotics
\begin{equation} \label{exp:phi0lo}
 \phi_0=\ln(a y^2)-\ln 2-1+2\gamma+\Psi\lb 1-\frac{\lambda}{2a} \rb-\frac{4}{\lambda}\frac{1}{y^2}+\frac{a}{\lambda}+\O{a y^2\ln(a y^2)},
\end{equation}
where, recall, $\Psi$ is the digamma function and $\gamma =-\Psi(1)=0.577216\ldots$ is the Euler-Mascheroni constant.

Expressions \eqref{philambdalo} and \eqref{exp:phi0lo} and the observation that $|f| \le (a y^2)^{-1} (a y^2 +1)^{\frac{\lambda}{2a}}e^{-\frac{a}{4} y^2}\|f\|_o$ give \eqref{eqn:phiouter}-\eqref{eqn:remouter}.
\end{proof}

%

\section{Matching and Eigenvalues} \label{matching}
In this section we prove the main result stated in the introduction.  We have two expressions, the inner and outer solutions, \eqref{phiinner} and \eqref{eqn:phiouter}, which solve ${\cal L}\phi_\lambda=\lambda\phi_\lambda$ in the common region $[R_0,R_i]$.
 The inner and outer solutions, \eqref{phiinner} and \eqref{eqn:phiouter}, should be equal, up to a constant multiple,  in the common region $[R_0,R_i]$.  Hence we require that
\begin{multline}
-\frac{1}{4}\lambda\ln y^2+ \frac{1}{y^2}+\frac{1}{4}(2\lambda+\mu)+{\cal R}_i \\ 
 =C\lb \ln y^2-\frac{4}{\lambda}\frac{1}{y^2}-\ln\frac{1}{a}-\ln 2-1+2\gamma+\frac{a}{\lambda}+\Psi\lb 1-\frac{\lambda}{2a}\rb +{\cal R}_o\rb
\label{eqn:prematching}
\end{multline}
for $y\in[R_o,R_i]$. Here ${\cal R}_i$ and ${\cal R}_o$ given in \eqref{reminner} and \eqref{eqn:remouter}. 

Note that the remainders in \eqref{phiinner} and \eqref{eqn:phiouter} are much smaller than the corresponding leading terms, i. e. $| {\cal R}_i|\ll a$ and  $| {\cal R}_o|\ll 1$, if
\begin{equation} \label{condio}
y \gg \ln\frac{1}{a}\sqrt{\ln R_o}\varepsilon_i R_i\ \textrm{and}\ ay^2 \ll 1,
\end{equation}
respectively. 
 We assume that $a\rightarrow 0$ and that \eqref{condio}  holds. 
Equating the leading terms in Eqn \eqref{eqn:prematching}, i.e. the terms which multiples of  $\ln y^2$, gives  the equation 
\begin{align*}
 C=-\frac{\lambda}{4}+R_C,
\end{align*}
with $|R_C| \lesssim \inf \frac{a|{\cal R}_o| +|{\cal R}_i |}{\ln y} $, where $\inf$ is taken over $R_o \le y \le R_i$ satisfying  the condition \eqref{condio}, and  
 therefore   $R_C =\O{a^{3/2}} $. (Here and below we take $\epsin\sim \epsout\sim a^{1/4}$.) Similarly, equating the constant terms in \eqref{eqn:prematching} and substituting the above expression for $C$ gives that
\begin{equation*}
2\lambda+\mu=\lambda\lb \ln \frac{1}{a}+\ln2+1-2\gamma-\frac{a}{\lambda} -\Psi\lb 1-\frac{\lambda}{2a} \rb\rb+R,
\end{equation*}
where the higher order term $R$ is bounded as $|R| \lesssim \inf (a|{\cal R}_o| +|{\cal R}_i |) $ and therefore satisfies
\begin{equation*}
R=\O{a^{3/2}\ln\lb\frac{1}{a}\rb}.
\end{equation*}
Rearranging the above equation, assuming that $\mu+a\ne 0$, gives that
\begin{equation*}
1=\frac{\lambda}{\mu+a}\lsb \ln\frac{1}{a}-\Psi\lb 1-\frac{\lambda}{2a} \rb+K \rsb+\frac{R}{\mu+a},
\end{equation*}
where, recall, $K:=\ln 2-1-2\gamma$. This is the equation \eqref{eqn:eigenvalue}. 
This proves Theorem \ref{thm:main}. \qed
\section{Solution of (\ref{eveqn})} \label{solving}
\begin{prop}
\label{lemma:eigenvalues}
  The set of solutions to \eqref{eveqn} 
   as $a\rightarrow 0$ is $\{\lambda_n\}_{n=0}^\infty$, where $\lambda_n$ is given by \eqref{eqn:explicit_eigvalue_estimates}.
\end{prop}
\begin{proof}
 The term $\ln\frac{1}{a}$ on the left hand side of \eqref{eveqn} is unbounded as $a\rightarrow 0$ whereas the right hand side is bounded.  Thus, there are two possibilities: either $\lambda/(\mu+a)\ll 1$ as $a\rightarrow 0$ or $|\lambda/(a+\mu)|\ge C>0$ and there is cancelation between $\ln\frac{1}{a}$ and $\Psi(1-\lambda/2a)$.

We begin with the first case.  If $\lambda/(\mu+a)\ll 1$ as $a\rightarrow 0$, then, since $\mu\lesssim a$, $\lambda/2a\sim \lambda/(\mu+a)$ as $a\rightarrow 0$.   We use this and the fact that $\Psi(1)=-\gamma$ and $\Psi(1+\delta)=\Psi(1)+\O{\delta}$ to write \eqref{eveqn} as
\begin{equation}
\frac{\lambda}{\mu+a}\lsb\ln\frac{1}{a}+K+\gamma +\O{\frac{\lambda \sqrt{a}}{\mu+a}}\rsb=1.
\label{lambda0est1}
\end{equation}
This equation immediately gives the rough estimate that $\lambda/(\mu+a)=\O{\ln^{-1}\frac{1}{a}}$.  Substituting this estimate into the $\O{\cdot}$ term in \eqref{lambda0est1}, then solving the resulting equation for $\lambda$ gives that
\begin{equation*}
 \lambda=\frac{\mu+a}{\ln\frac{1}{a}+K+\gamma}\frac{1}{1+\O{\ln^{-2}\frac{1}{a}}}.
\end{equation*}
Further simplification of the right hand side gives the $n=0$ expression in  \eqref{eqn:explicit_eigvalue_estimates}.

If $|\lambda/(\mu+a)|\ge C>0$, then there must be cancelation between $\ln\frac{1}{a}$ and $\Psi\lb 1-\frac{\lambda}{2a}\rb$ in \eqref{eveqn} as already mentioned above.  The digamma function $\Psi(x)$ has poles at $x=-n$ for integers $n\ge 0$ and hence, for cancelation to occur, $\lambda/2a$ must have the form
\begin{equation}
 \frac{\lambda}{2a}=1+n+\delta,
\label{eqn:lambda2a}
\end{equation}
where $\delta\ll 1$ as $a\rightarrow 0$.  Substituting this form of $\lambda/2a$ into \eqref{eqn:eigenvalue} gives the equation
\begin{equation}
 (1+n+\delta)\lb \ln\frac{1}{a}-\Psi(-n-\delta)+K \rb=\frac{\mu+a}{2a}.
\label{eveqnnge1}
\end{equation}
We extract the singular behaviour of $\Psi(-n-\delta)$ using the identity
\begin{equation}
 \Psi(-n-\epsilon)=\Psi(1-\delta)+\sum_{k=0}^{n}\frac{1}{k+\delta}.
\label{PsiIdentity}
\end{equation}
If $k\ge 1$ and $\delta<1/2$, then $1/(k+\delta)\lesssim 1/k+\delta/k^2$.  Using this bound and the fact that $\Psi(1-\delta)=-\gamma+\O{\delta}$ in \eqref{PsiIdentity} yields that
\begin{equation*}
 \Psi(-n-\delta)=\frac{1}{\delta}-\gamma+\sum_{k=1}^n\frac{1}{k}+\O{\delta}.
\end{equation*}
Substituting the right hand side for $\Psi(-n+\delta)$ in \eqref{eveqnnge1} gives the equation
\begin{equation*}
 (1+n+\delta)\lb \ln\frac{1}{a}-\frac{1}{\delta}+K+\gamma-\sum_{k=1}^n\frac{1}{k}+\O{\delta} \rb=\frac{\mu+a}{2a}.
\end{equation*}
As before, a rough estimate of $\delta$ is $\ln^{-1}\frac{1}{a}$.  Using this to invert $(1+n+\delta)$ gives that
\begin{equation*}
 \ln\frac{1}{a}-\frac{1}{\delta}+K+\gamma-\sum_{k=1}^n\frac{1}{k}+\O{\ln^{-1}\frac{1}{a}}
 =\frac{\mu+a}{2a(1+n)}\frac{1}{1+\O{\ln^{-1}\frac{1}{a}}}
\end{equation*}
and hence, solving this equation for $\delta$, we find that
\begin{equation*}
 \delta=\frac{1}{\ln\frac{1}{a}+K+\gamma-\sum_{k=1}^n\frac{1}{k}-\frac{\mu+a}{2a(n+1)}}
 +\O{\ln^{-3}\frac{1}{a}}.
\end{equation*}
Substituting this expression into \eqref{eqn:lambda2a} gives the eigenvalue approximation \eqref{eqn:explicit_eigvalue_estimates} for $n\ge 1$ (with $n$ replaced by $n-1$ above) and completes the proof.
\end{proof}

\section{Numerical Calculation of Spectrum} \label{numerics}

To determine eigenvalues and eigenfunctions numerically we used a version of shooting method: we numerically solved the eigenvalue equation \eqref{eveqn} with initial conditions $\phi_\lambda(y=0)=1$ and $\frac{d}{dy}\phi_\lambda(y=0)=0$ for each value of $\lambda$. For general $\lambda$, solution at $ay^2\gg 1$ is a linear combination (\ref{generalconfluenthypergeometric}), which grows exponentially as given by  (\ref{seriesinftyhypergemetric}). We used Newton's method to find values $\lambda$ for which $c_1=0$ (i.e. removing exponentially growing terms at infinity) in (\ref{generalconfluenthypergeometric}).
Stopping criterion for Newton's method was to have both   $\phi_\lambda(y)$ and $\frac{d}{dy}\phi_\lambda(y)$ to decay exponentially for large $y$. Note that vanishing of $\phi_\lambda(y)$ for $y\to \infty$ is not sufficient because it would not exclude spurious solution when $c_1\phi_0(y)+c_2\phi_1(y)=0$ at one point only. We controlled numerical precision by the matching of numerical solution to the asymptotics  (\ref{seriesinftyhypergemetric}).

  Figure 1 shows  the first three eigenvalues as functions of $a$ obtained in three different ways. The solid lines are numerically computed from shooting method, the dashed lines are obtained using the expressions in \eqref{eqn:explicit_eigvalue_estimates}, and the dash-dot lines (almost visually indistinguishable from solid lines) are obtained by \eqref{eveqn}.
 It is seen that accuracy of numerical solutions compare with approximate analytical results is very high.

\section{Conclusions}\label{sec:concl}

We summarize main results of this paper:
\begin{itemize}
\item We found low-lying spectrum for a class of operators which appear in the linearization of the simplified critical Keller - Segel system around the one-parameter family of stationary solutions. These operators have one negative and one near zero eigenvalue and as a result - as discussed in the introduction - the blowup asymptotics will be governed by a two-parameter deformation of the static solution.  The construction of such deformations and ensuing results are outlined in the introduction.

\item 
We constructed a rigorous and robust version of 
the method of matched asymptotics. We believe it 
can be used a large variety of linear differential operators arising in the linearization of nonlinear equations and hopefully can be extended to nonlinear ones as well (in this case perturbation series, \eqref{inNeumann} and \eqref{oseries}, should be replaced by fixed point iterations).
\end{itemize}

There are two main limitations of our set-up: we deal with radially symmetric solutions and with adiabatic approximation ignoring evolution of chemical concentration. Hence we conclude by emphasizing the desirability of two further extensions of our analysis by considering
\begin{itemize}
\item non-radially symmetric initial conditions;
\item the full  Keller - Segel model (without the adiabatic approximation).
\end{itemize}

\section{Acknowledgments}

The authors are grateful to the anonymous referee and M. Brenner for useful remarks.

Work of P.L. was supported by NSF grants DMS 0719895,
DMS 0807131 and UNM RAC grant. Research of Yu. Ov. is supported under the grant 
RFBR-07-0-2-12058 and research of  I.M.S. is supported by NSERC under Grant NA 7901.

A part of this work was done during I.M.S.'s stay at the IAS, Princeton, and ESI, Vienna.


\appendix

\section{Appendix: Solutions of ${\cal L}_0\phi=0$} 
\label{App:eqnL0phi=0}

 In this appendix we derive, using well-known properties of  the  confluent hypergeometric functions, some properties of solutions of  equation $\Lopout\phi=0$ which were used in Section \ref{sec:outer}.
\begin{prop} \label{prop:phiibnd}
If $\lambda/2a\ne$positive integer, then there are two independent solutions, $\phi_0$ and $\phi_1$,  of the equation ${\cal L}_0\phi=0$, satisfying  the bounds \eqref{phi0bnd} and \eqref{phi1bnd},
and having the Wronskian \eqref{eqn:wronskian1} and the expansion \eqref{exp:phi0lo}.
\end{prop}
\begin{proof}
The equation $\Lopout\phi=0$ is the eigenvalue equation for the spherically symmetric harmonic oscillator in $D=4$:
\begin{equation}\label{Lharmonic}
\Lopout\phi =\left [-\frac{1}{y^3}\partial_y y^3\partial_y
+\frac{a^2}{4}y^2-\lambda \right ] \phi=0.
\end{equation}
Changing in this equation  both dependent and independent variables,
\begin{equation} \label{changevarhypergeometric}
\phi(y)=e^{-ay^2/4}\chi(z),  \quad z=\frac{ay^2}{2},
\end{equation}
we obtain the Kummer's (or a confluent hypergeometric) differential equation \cite{AbramowitzStegun}
\begin{equation} \label{confluenthypergeometric}
z\chi''+(2-z)\chi'+\lb\frac{\lambda}{2a}-1\rb\chi=0.
\end{equation}
Assuming that  $\lambda/2a\ne$positive integer, the latter equation has two linearly independent solutions 
\begin{equation} \label{generalconfluenthypergeometric}
\chi_0= U(1-\frac{\lambda}{2a},2,z),\ \chi_1=_1F_1(1-\frac{\lambda}{2a},2,z),
\end{equation}
where $U(a,b,z)$ and  $_1F_1(a,b,z)$ are the  confluent hypergeometric functions of the second kind and the first kind, respectively \cite{AbramowitzStegun}. They are
given by
 the following expressions (see the equations (13.1.2), (13.1.6) in \cite{AbramowitzStegun}  and the equation (13) of the section 6.7.1 of \cite{BatemanErdelyi1953}):
\begin{equation} \label{phi01seriesfull}
\begin{split}
_1F_1(a,b,z)&=\sum\limits^{\infty}_{r=0}\frac{(a)_r z^r}{(b)_r r!},\\
 U(a,n+1,z)&=\frac{(-1)^{n+1}}{n!\Gamma(a-n)}\Big  [ \, _1F_1(a,n+1,z)\ln{z} \\
&+\sum\limits_{r=0}^{\infty} \frac{(a)_r z^r}{(n+1)_r r!}\{ \Psi(a+r)-\Psi(1+r)-\Psi(1+n+r)   \} \Big ]\\
&+\frac{(n-1)!}{\Gamma(a)}\sum\limits_{r=0}^{n-1}\frac{(a-n)_r}{(1-n)_r} \frac{z^{r-n}}{r!} ,
\end{split}
\end{equation}
where $(a)_j=a(a+1)(a+2)\ldots (a+j-1), \ (a)_0=1$, $n=0, \ 1, \ 2, \ldots$,  and the principal branch of $\ln z$ is assumed to be chosen by setting  $-\pi < \arg{z} \le \pi.$

The equations (\ref{changevarhypergeometric}) and (\ref{generalconfluenthypergeometric}) give two linearly independent solutions $\phi_0$ and $\phi_1$ of $\Lopout\phi=0$  on $[0,\infty)$:
\begin{equation} \label{phi01sol}
\begin{split}
\phi_0 (y)&=\Gamma(-\frac{\lambda}{2a})U(-\frac{\lambda}{2a}+1,2,\frac{ay^2}{2})e^{-ay^2/4}, \\
\phi_1 (y)&=_1F_1(-\frac{\lambda}{2a}+1,2,\frac{ay^2}{2})e^{-ay^2/4},
\end{split}
\end{equation}
where, using that $\phi_0$ and $\phi_1$  are defined up to arbitrary constants, we  added,  for convenience, the factor  $\Gamma(-\frac{\lambda}{2a})$. (Without that factor the equation  (\ref{exp:phi0lo})   would have a factor $1/\Gamma(-\frac{\lambda}{2a})$ in all terms.)

Asymptotic expansions of both confluent hypergeometric functions in (\ref{generalconfluenthypergeometric}) for $a y^2\to \infty$ are given by (see the equations (13.1.4) and (13.1.8) of Ref. \cite{AbramowitzStegun})
\begin{equation} \label{seriesinftyhypergemetric}
\begin{split}
\phi_0 (y)&=\Gamma(-\frac{\lambda}{2a})\left (\frac{ay^2}{2}\right )^{\frac{\lambda}{2a}-1}e^{-ay^2/4}\left [1+O\left(\frac{1}{ay^2}  \right )   \right ],\\
\phi_1 (y)&=\frac{1}{\Gamma(1-\frac{\lambda}{2a})}\left (\frac{ay^2}{2}\right )^{-\frac{\lambda}{2a}-1}e^{ay^2/4}\left [1+O\left(\frac{1} {ay^2} \right ) \right ].
\end{split}
\end{equation}
It follows from equations (\ref{seriesinftyhypergemetric}) that out of these two functions only  $\phi_0$ decays at infinity (with correct asymptotic $\phi_0 \propto y^{\frac{\lambda}{a}-2}e^{-\frac{a}{4} y^2}$).

The bounds \eqref{phi0bnd} and \eqref{phi1bnd} are proven similarly, so we prove only  \eqref{phi0bnd}. Recall the definition of $\chi_0$ in \eqref{generalconfluenthypergeometric}.  \eqref{phi01seriesfull} implies that for a constant $C_1$, which might depend only on $\lambda/2a$, there is a point $z_0 \ge 1$ s.t. $|\chi_0 (z)| \le C_1 z^{\frac{\lambda}{2a}-1} $ for all $z \ge z_0$. Since $\chi_0$ depends only on $z$ and $\lambda/2a$, there is a  constant $C_2$, which might depend only on $\lambda/2a$, s.t. $|\chi_0 (z)| \le C_2 z^{\frac{\lambda}{2a}-1}$ for all $z \ge 1$.

Now, \eqref{phi01seriesfull} implies that $|\chi_0 (z)| \le C_3/z$ for all $0 \le z \le 1$ for some absolute constant $C_3$. Combining the above inequalities gives the bound \eqref{phi0bnd}.

Now we compute the Wronskian, $W(\phi_0,\phi_1):=\phi_0\p_y\phi_1-\p_y\phi_0\phi_1$, of the two solutions found above. As usual, using the equation ${\cal L}_0\phi=0$, we derive the first order equation, $W' = -\frac{3}{y}W$, for $W(\phi_0,\phi_1)$.  Solving this equation, we obtain
$
W(\phi_0,\phi_1)=\frac{C}{y^3},
$
with a real constant $C$.  To find this constant we compute $W$ as $y \rightarrow \infty$, using the equations \eqref{seriesinftyhypergemetric} and the fact $\Gamma(z+1)=z\Gamma(z)$ (on can also find $C$ using the expansion (\ref{phi01seriesfull})). This gives $C=-\frac{8}{\lambda}$ and \eqref{eqn:wronskian1}.

Finally, we prove \eqref{exp:phi0lo} for $\phi_0$.  To this end we study the behaviour of the solution for $y\ll 1/a^\frac{1}{2}\ln^\frac{1}{2}\frac{1}{a}$, $a\ll 1$, or equivalently, $z\ll 1/\ln\frac{1}{a}$.  In the small $z$ region,
\begin{equation*}
 \chi_{\lambda}(z)=\ln \frac{z}{A}-\frac{2a}{\lambda z}+\O{z}
\end{equation*}
and $e^{-\frac{1}{2}z}=1-\frac{1}{2} z+\O{z^2}$.  Computing the product of the small $z$ expansions of $\chi_\lambda$ and $e^{-\frac{z}{2}}$, and replacing $z$ with $\frac{a y^2}{2}$ in the result gives the expression \eqref{exp:phi0lo}.

Now let $y\in (0,\Rin]$. Using the expansion  (\ref{phi01seriesfull})   for $n=1$ and (\ref{changevarhypergeometric}) we obtain    \eqref{exp:phi0lo}.
\end{proof}




\end{document}